\pdfoutput=1
\documentclass[showpacs,superscriptaddress,aps,pra, twocolumn]{revtex4-1}
\usepackage{caption}
\usepackage{subcaption}
\usepackage{cancel}
\usepackage{amsfonts}
\usepackage{amsmath,amsthm}
\usepackage{bm}
\usepackage{amssymb}
\usepackage{mathrsfs}
\usepackage{amsmath, amsthm, amssymb,amsfonts,mathbbol,amstext}
\usepackage{mathtools}
\usepackage[colorlinks=true,citecolor=blue,urlcolor=blue]{hyperref}
\usepackage[normalem]{ulem}
\usepackage{mathrsfs}
\usepackage{dsfont}
\usepackage{graphicx}

\newtheorem{dfn}{Definition}

\newtheorem{thm}{Theorem}

\newtheorem{prop}[thm]{Proposition}
\newtheorem{result}[thm]{Result}

\newcommand{\ket}[1]{\left| #1 \right\rangle}

\newcommand{\ketbra}[2]{\left|#1\middle\rangle\middle\langle#2\right|}

\newcommand{\mc}[1]{\mathcal{#1}}
\newcommand{\mds}[1]{\mathds{#1}}

\def\be{\begin{equation}}
\def\ee{\end{equation}}
\usepackage{color}
\definecolor{violeta}{cmyk}{0.07,0.90,0,0.34}

\newcommand{\cond}[2]{\mds{P}(#1 \vert #2)} 

\usepackage{cancel}
\usepackage{soul}
\definecolor{cgreen}{RGB}{26, 199, 76}

\begin{document}


\title{Effective state as compatibility between agents}

\author{Cristhiano Duarte} 
\affiliation{Schmid College of Science and Technology, Chapman University, One 
University Drive, Orange, CA, 92866, USA}

\begin{abstract}

Shedding a new light in the coarse-graining scenario, in this contribution we came up with different necessary and sufficient conditions for the existence of a well-defined coarse-grained state. For doing so, we had to break apart with the usual quantum channels perspective and assume a more decision-theoretical posture. Broadly speaking, we reinterpret the coarse-graining problem in the language of quantum state pooling, and by make an extensive use of the conditional quantum states toolkit we have been able to derive more tangible conditions for the emergence of a compatible effective coarse-grained state. 

\end{abstract}

\maketitle


\section{Introduction} \label{Sec.Intro}

Effective descriptions of complex systems are spread everywhere we look at. From the very known cases of thermodynamics~\cite{Callen} and statistical physics~\cite{Castiglione08,Wolfram83}, passing through biological and chemical models~\cite{DPSBJ92,MichalEtAl15,PSP13,BaronEtAl07,LW75}, and going all the way up to engineering and economics problems~\cite{BeiEtAl16,BS78,Ilachinski02,PST00}, whenever we are faced with a highly detailed, intricate system or process we naturally try to come up with a simpler, coarse-grained version for it.  A version for which we can still extract meaningful physical properties.

Take for instance the case of a heat machine, say a steam engine. We do not have to solve loads and loads of Schroedinger equations every time we want to assemble a train engine. We do not have to take into account the quantum microscopic details to talk about how such an engine works~\cite{VWS86}. We are satisfied with the very coarse-grained description provided by thermodynamics, we can even foresee that due to the heat released by the train wheels, the trails may heat up, expand, and if not otherwise consider break apart causing an accident. That would be prohibitive if such coarse-grained description were not the case. 

Although we have been using simpler descriptions for complex systems, there has not been a mathematically rigorous prescription taking care of how coarse-graining descriptions of quantum systems, particularly quantum dynamics, end up matching with our classical description of reality. It was to fill this gap that the authors of Ref.~\cite{DCBdM17} proposed their toy model. 

Overall, the model set up in Ref.~\cite{DCBdM17} and further explored in ~\cite{PdM19,PdM19_02}  is a mathematical way to investigate the possible appearing of an effective dynamics arising from the lack, blur, or misrepresentation of information about the underlying microscopic system evolving through a closed, unitary dynamics $\mc{U}$. Making things precise, the authors modeled the lost of information as a $CPTP$ map~\cite{NC00,MichaelGuide}, and asked for necessary and sufficient conditions for the existence of quantum channel $\Gamma$ to be seen as the emergent map. The only constraint being that $\Gamma$ should be \emph{compatible} with the diagram describing the scenario (see Fig.~\ref{Fig_Diagram_Quantum}).  

The compatibility they demand is algebraic and must be seen as expressing the idea that coarse-grained descriptions must still be meaningful and representative of the system one wants to center attention on. In a nutshell, for them a quantum channel $\Gamma$ is an effective description, whenever the diagram in Fig.~\ref{Fig_Diagram_Quantum} commutes.In here we will address the problem in a slight different manner, though.

Although we keep up with a diagrammatic approach resembling that of Ref.~\cite{DCBdM17}, our contribution switches from the algebraic commutativity demand to a more decision-theoretical requisition. Interpreting quantum states as degree of beliefs, or information an agent has about a system of interest, we will use the (quantum) state pooling task~\cite{LS14} to frame in it the usual coarse-graining scenario. It is exactly this changing of picture that allows us to obtain more tangible necessary and sufficient conditions for the existence of an accurate coarse-grained state.

For doing so we also needed to change the way we use the quantum formalism. Aligned with the perspective put out by the authors of Refs.~\cite{LS13,LS14,LP08}, and regarding quantum theory as a framework for Bayesian inference, we make an extensive use of the conditional quantum states formalism. It is within the intersection of decision-theoretical tasks with the conditional quantum state approach that our contribution fits in. Our main result is a clear example of it.
 
We have subdivided the paper as follows: in Sec.~\ref{Sec.Preliminaries} we have put together all the information we feel necessary for understanding our main results. The scenario we want to describe is reviewed in Subsec.~\ref{Subsec.CG}, the mathematical formalism we will use as a tool is introduced in Subsec.~\ref{Subsec.ConditionalQS}, and finally Subsec.~\ref{Subsec.Agreement} contains the decision-theoretical definitions and results we will use immediately in the next section. Particularly, Subsec.~\ref{Subsec.Agreement} makes a parallel between classical and quantum definitions, and although we will only use the quantum version of them, we think it is more didactic to take this route as it creates a better environment for learning new concepts we usually are not exposed to in the field. Moving on, rewriting the coarse-graining scenario as a decision-theoretical task, Sec.~\ref{Sec.MainResults} provides our main results. Sec.~\ref{Sec.Conclusion} wraps this paper up and contains our conclusions and discussions as well as possible further works.

\section{Preliminaries}\label{Sec.Preliminaries}


\subsection{Coarse-Graining and effective dynamics}\label{Subsec.CG}

Inspired by the usual idea that even our best macroscopic description for the world surrounding us is nothing but an emergent, coarse-grained portrayal of what happens in the underlying microscopic reality~\cite{Callen,Kofler1,Castiglione08,GellMann14,FaistThesis16, Leifer14,ZurekClassToQuant02,Zurek09,BKZ06,BPH15,Spekkens07},  the authors of Ref.~\cite{DCBdM17} came up with a mathematically well defined, toy model exploring the rising of emergent, coarse-grained dynamics. 

Although simple in its inner details, their model was supposed to capture the main idea we bear in mind when facing with collective, complex phenomena in the sciences. The idea that in the absence of a full description for a system, we ought to be able to provide an effective one. Effective in the sense that when ignoring part of the details we cannot keep track of it is still meaningful, and also brings over some sort of information.

 Think of the usual undergrad thermodynamics, where even though we do not know how to diagonalise the Hamiltonian for a gas of interacting molecules, we still know --by coarse-graining some details--  how to build up trains, cars, steam-engines and extract useful work from them. This permanent loss of underlying details have not hindered us of doing science, rather we have learned how to deal with it phenomenologically~\cite{Callen}. 
 
 The work in Ref.~\cite{DCBdM17} was built upon these ideas, and tried to frame them in a sort of a quantum-to-classical paradigm. The main point there being that of coming up with a macroscopic emergent description consistent, compatible with the microscopic underlying unitary closed dynamics. Compatibility between macroscopic and microscopic descriptions was their key point and it is this very concept we want to explore from a different perspective here. To let clear the distinction between what we mean and they meant by compatible descriptions it is worthy to quickly revisit their main arguments.        

Initially assuming the lacking-of-details, or blurring of information, as given by a completely positive trace preserving (CPTP) map $\Lambda:\mc{D}(\mc{H}_{D}) \rightarrow \mc{D}(\mc{H}_{d})$, and the underlying microscopic dynamics given by a unitary $\mc{U}_{t}:\mc{D}(\mc{H}_{D}) \rightarrow \mc{D}(\mc{H}_{D})$, the formalism developed in Ref.~\cite{DCBdM17} looked for the actual form (when well-defined) assumed by the emergent dynamics $\Gamma_{t}:\mc{D}(\mc{H}_{d}) \rightarrow \mc{D}(\mc{H}_{d})$ such that the diagram depicted in Fig.~\ref{Fig_Diagram_Quantum} commutes. The latter meaning that the emergent dynamics $\Gamma_{t}$ to be assigned must obey the idea that experiencing first the underlying $\mc{U}_{t}$ followed by $\Lambda$ is the same that first experiencing the coarse-graining and only then $\Gamma_{t}$. 

In either case, these descriptions must agree one another and output the same final state, as they represent the same dynamics. The commutativity of the diagram in Fig~\ref{Fig_Diagram_Quantum} responsible by their notion of \emph{compatibility}.

\begin{center}
    \begin{figure}
        \includegraphics[scale=0.2]{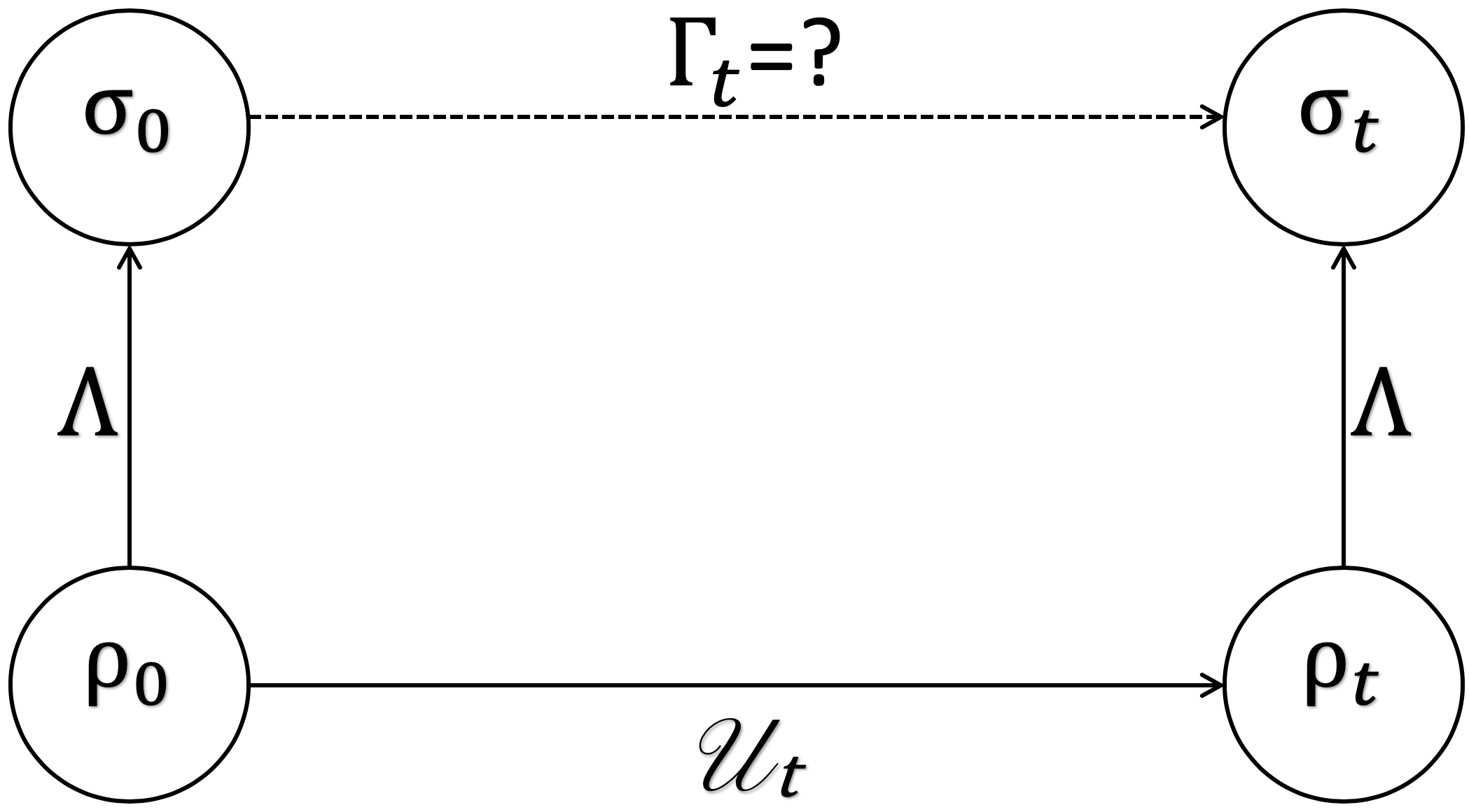}
        \caption{\begin{footnotesize}Coarse-graining diagram. Vertical arrows represent loss, lacking or blurring of information. Lower horizontal arrow represents the closed, unitary dynamics the system is going through. Uppermost horizontal arrow frames the emergent, perceptible, macroscopic dynamics.\end{footnotesize} }
        \label{Fig_Diagram_Quantum}
    \end{figure}
\end{center}

This brief exposition about the coarse-graining  scenario is more than enough for our purposes. For more details, though, including discussions on the existence of emergent maps, the reader should check out references~\cite{DCBdM17} and~\cite{SRS18}. More concrete examples,  also going along the lines of the aforementioned reference, can be found in Refs.~\cite{SRFCMK15,PdM19,PdM19_02}. We shall now definitively move on to the next topic.


\subsection{A glimpse on quantum conditional states}\label{Subsec.ConditionalQS}

We dedicate this bit for a quite brief introduction to the formalism of quantum conditional states. This section should be thought of as a short compilation of definitions and results involving a quantum generalization for the usual notion of conditional probability. Consequently, it is not our intention to explore in-depth the formalism developed for the authors in Refs.~\cite{LS13,LS14,LP08}, rather, we only wanted to provide to the reader few definitions we will make extensive use in the subsequent sections. 

Kicking this section off, we should say we adopt here the philosophy put out by the authors in Ref.~\cite{LS14}, that of seeing the conditional states formalism as nothing but a tentative to treat quantum theory as a generalization of the classical theory of (Bayesian) inference. Emphasizing, particularly the causal neutrality of the latter. Overall, this is the parallel we are also trying to draw here, and it is the very notion of (Bayesian) conditioning that will allows us to re-frame the coarse-scenario into the picture of decision theory.

In the usual classical setting the most basic object is nothing but a joint probability distribution $\mds{P}(X,Y,...,Z)$ describing an agent's degrees of belief,  knowledge, or even information about a list of classical random variables $X,Y,...,Z$. The latter might represent different properties of a system in a given instant of time, the same property of a system at different times, or even more elaborate, unusual things like a list of mathematics axioms an agent wants to assume as being either true or false. The fact of the matter being that there is no interpretative constraints on what those classical random variables in the list mean.

 Seeking to achieve the same level of interpretative freedom, whereas in the usual texts of quantum mechanics~\cite{cohen77,peres95,NC00,MichaelGuide} the most basic unit is the concept of a \emph{quantum system}, in the quantum conditional states perspective~\cite{LS13,LS14} the most basic concept is that of a \emph{region}. Broadly, an elementary region is anything we want to model as quantum, and attach to a Hilbert space. For example, while in the conventional quantum approach the input and output of a quantum channel are treated as being the same system at two different instants of time, in the conditional states formalism we would associate to the same input and output two distinct elementary regions: $\mc{H}_{\mbox{in}}$ and $\mc{H}_{\mbox{out}}$, respectively. The composed region being  the tensor product $\mc{H}_{\mbox{in}} \otimes \mc{H}_{\mbox{out}}$ of these elementary factors. We highly recommend Ref.~\cite{LS13} for an in-depth, detailed explication.

Moving on, the parallel we want to set up is basically the following: within the conditional states formalism classical random variables become quantum systems, and probability distributions become trace class operators --not necessarily positive-- acting on Hilbert spaces. For each elementary region $A$ we associate a Hilbert space $\mc{H}_A$. For a composite region $AB$, with $\mc{H}_{A}$ attached to $A$ and $\mc{H}_{B}$ to $B$, we associate the tensor product $  \mc{H}_{AB}:=\mc{H}_{A} \otimes \mc{H}_{B}$. The information, knowledge, or degree of belief about $AB$ being described by an operator $\sigma_{AB}$ on $\mc{H}_{AB}$. Partial information coming through a composed description $\sigma_{AB}$ is handled as it is in the classical case. Whereas in the latter we marginalize summing over the outcomes of a random variable 
\begin{align}
\mds{P}(X)=\sum_{y \in \mbox{Out}(Y)}\mds{P}(X,Y=y),
\end{align}
in the conditional states formalism we trace out the system we do not want to take care of:
\begin{align}
\sigma_{A}=\mbox{Tr}_{B}\left( \sigma_{AB} \right) \in \mc{D}(\mc{H}_{A}).
\end{align}
Finally, we demand that whenever $A$ is meant to be a elementary region, any state associated with it must be positive, Mathematically: 
\begin{align}
\sigma_{A} \geq 0, \,\, \mbox{if $A$ is an elementary region.} 
\label{Eq.DefPositivityElementaryRegion}
\end{align}

This in turn implies that if $\sigma_{AB}$ is the state associated with the composition of two elementary regions $A$ and $B$, each marginal state $\sigma_{A}$ and $\sigma_{B}$ must be a positive operator, although $\sigma_{AB}$ need not to be. 

Another important aspect of classical probability theory also meets its counterpart in the formalism of conditional states. While classically, given two classical variables $X,Y$ the probability of $X$ conditioned, or given, $Y$ is defined via
\begin{align}
\mds{P}(X=x \vert Y=y):= \frac{\mds{P}(X=x,Y=y)}{\mds{P}(Y=y)},
\label{Eq.DefCondProbClassical}
\end{align}
in here we define:
\begin{align}
\sigma_{B \vert A}:= \sigma_{AB} \star \sigma_{A}^{-1},
\label{Eq.DefCondProbClassical}
\end{align}
as for the \emph{conditional state of $B$ given $A$}, where the $\star-$product is a non-commutative operation deeply explored in Refs.~\cite{LS14,LS13,LP08} and defined as:
\begin{align}
\Psi_{AB} \star \Phi_{B} := (\mds{1}_{A} \otimes \Phi_{B})^{\frac{1}{2}} \Psi_{AB}(\mds{1}_{A} \otimes \Phi_{B})^{\frac{1}{2}}.
\label{Eq.DefStarProduct}
\end{align}

Although there are some limitations~\cite{LS13,LS14}, the conditional states formalism is robust enough to handle with correlations among classical and quantum regions simultaneously. We wrap up this section describing how to do it through what is called \emph{hybrid states}. 

For a classical variable $X$ we want to consider in conjunction with a quantum region $B$, we associate a Hilbert space $\mc{H}_{X}$ together with a preferred basis $\{ \ket{x} \}_{x \in \mbox{Out}(X)}$. On the other hand, $B$ is bounded to a Hilbert space $\mc{H}_{B}$ with no preferred structure. The composed region being $\mc{H}_{XB}=\mc{H}_{X} \otimes \mc{H}_{B}$. Additionally, we want to ensure that the classical region remains classical, so that not only there must not exist entanglement between $X$ and $B$ but also the reduced state shall be diagonal in the preferred basis. Respecting these conditions, we define the \emph{hybrid state} between $X$ and $B$ to be:
\begin{align}
\sigma_{XB}:=\sum_{x \in \mbox{Out}(X)} \ketbra{x}{x} \otimes \sigma_{X=x,B},
\label{Eq.DefHybridState}
\end{align}
where each  $\sigma_{X=x,B}$ is a positive operator acting on $\mc{H}_{A}$. We refer to Ref.~\cite{LS13} for interesting examples as well as an in-depth discussion of the limitations imposed on the conditional states approach.


\subsection{Compatibility and Pooling in a Nutshell}\label{Subsec.Agreement}

This subsection brings in and explores only that definitions we need from decision theory. As a matter of fact, we will restrict ourselves only to that material that will be extensively used in the next section in order to get new light on the coarse-graining scenario. Ultimately, our main goal is to take seriously and explore deeper that idea that the lowermost path and the uppermost one (see Fig.~\ref{Fig_Diagram_Quantum}) must be compatible, somehow agreeing one another.  

For the authors of Ref.~\cite{DCBdM17}, compatibility was expressed by demanding that the diagram depicted in Fig.~\ref{Fig_Diagram_Quantum} commuted. Here, on the other hand, instead of adopting such algebraic perspective, we will frame the same problem within a different picture and putting ourselves in the shoes of a Bayesian decision-maker, we will show how the \emph{state pooling} task might also be useful to study emergent behaviours~\cite{Zurek09,BKZ06,BPH15}. 

Adopting the general methodology to combine states as developed in Ref.~\cite{LS14}, and rather than trying to come up with a universal, valid rule for all situations, we will base our analysis on the very well known and down-to-earth concept of Bayesian conditioning~\cite{sep-epistemology-bayesian}. Within this framework it is possible to come up with quite natural definitions~\cite{LS14} for what compatibility between probability assignment means. This kind of argument is robust and does not rely on \emph{ad-hoc} methods to decide whether assignments might be regarded as being compatible. Summing up, Bayesian conditioning seems to be the best approach one might think of when putting the coarse-graining scenario into the shoes of decision theory~\cite{Anand95}.

For our purposes one particular task in decision theory stands out, that of \emph{state pooling}. In a nutshell, it is the task~\cite{LS14,French80,GZ86,Jacobs95} of combining different states assignments into a single assignment accurately representing the beliefs, information, or knowledge of the involved group of agents as a whole. In practice, state pooling is the kind of problem our politicians should face on a daily basis, as decisions~\cite{Morris74, Herbut04} are ought to be made as a group and conflict of preferences prevent all involved parts in the group from maximize their objectives, or face the world with their very own beliefs, simultaneously. 

\subsubsection{Compatibility}\label{Subsub.Compatibility}

Before rigorously defining what we mean by state pooling, we need first define what we mean by compatibility, or agreement. We will go through it superficially over the following lines, but for a deeper, detailed approach to the same topic we refer to~\cite{LS14}. In a nutshell, depending on the Bayesian perspective one might adopt about probability theory, there exist two equivalent notions for classical compatibility:

\begin{dfn}[Objective Classical Compatibility]
 Two probability assignments $\mds{Q}_1(Y)$ and $\mds{Q}_2(Y)$, over a random variable $Y$, are \emph{compatible} whenever it is possible to find two random variables $X_1$ and $X_2$, a joint probability distribution $\mds{P}(Y,X_1,X_2)$, and two outcomes $x_1 \in \mbox{Out}(X_1)$ and $x_2 \in \mbox{Out}(X_2)$   such that:
 \begin{enumerate}
 \item $\mds{P}(X_1=x_1,X_2=x_2) > 0 $;
 \item $\mds{Q}_i(Y=y)=\cond{Y=y}{X_i=x_i}$, 
 \end{enumerate}
 for all $i \in \{1,2\}$ and for all $y \in \mbox{Out}(Y)$. 
 \label{Def.CompatibilityObjective}
\end{dfn}

\begin{dfn}[Subjective Quantum Compatibility]
Two probability assignments $\mds{Q}_1(Y)$ and $\mds{Q}_2(Y)$, over a random variable $Y$, are \emph{compatible} whenever it is possible to find another random variable $X$ and a conditional distribution $\cond{X}{Y}$ that both agents agrees upon such that:
\begin{enumerate}
\item $\sum_{y \in \mbox{Out}(Y)}\cond{X=x}{Y=y}\mds{Q}_i(Y=y) > 0$ for all $x$ and for all $i \in \{1,2\}$;
\item  $\mds{P}_1(Y \vert X=\tilde{x}) = \mds{P}_2(Y \vert X=\tilde{x})$, for some value $\tilde{x}$ in $\mbox{Out}(X)$, where
\begin{align}
\mds{P}_i(Y=y \vert X=x):=\frac{\cond{X=x}{Y=y}\mds{Q}_i(Y=y)}{\sum_{y}\cond{X=x}{Y=y}\mds{Q}_i(Y=y)}.
\label{Eq.DefSubjectiveCompPi}
\end{align}
\end{enumerate}
\label{Def.CompatibilitySubjective}
\end{dfn}

Although  Def.~\ref{Def.CompatibilityObjective} and Def.~\ref{Def.CompatibilitySubjective} must be seen within the scope of a Bayesian's perspective of probability theory, the former is objective-like in its form, in the sense that it reflects the idea that objective agents can only disagree about the assignment they attach to a system if they have had experienced different data (represented by $X_1$ and $X_2$) coming from the system of interest. The latter, though, fits into the subjective approach, and must be seen as saying that agents ought to agree to each other when they can jointly come up with a test (represented by $X$) and a data-set (represented by $\cond{X}{Y}$) such that there is agreement between the two agents when they look at their Bayesian update $\mds{P}_i(Y \vert X=\tilde{x}) $ for some particular value $\tilde{x}$.

The next theorem~\cite{LS14} connects the subjective and objective perspectives, and consequently also justifies our affirmation saying that they are nothing but faces of the same coin:

\begin{thm}
$\mds{Q}_{1}(Y)$ and $Q_{2}(Y)$ are compatible, either subjectively or objectively, if and only if 
\begin{align}
\mbox{supp}[\mds{Q}_1(Y)] \cap \mbox{supp}[\mds{Q}_2(Y)] \neq 0.
\label{Eq.ThmComptNecSuffClass}
\end{align}
\label{Thm.NecessaryAndSufficientConditionForCompClassical}
\end{thm}

Before delving even deeper into a more general version of compatibility, a little detour might be pedagogic here. Exemplifying how Def.~\ref{Def.CompatibilityObjective}, and consequently Def.~\ref{Def.CompatibilitySubjective}, work and are useful for deciding whether two distinct probability assignments are compatible, we state the two following results:

\begin{prop}
Let $Y$ be a binary random variable, say $\mbox{Out}(Y)=\{a,b\}$. Additionally, let $\mds{Q}_{1}(Y)=(p,1-p)$ and $\mds{Q}_{2}(Y)=(q,1-q)$ be two distinct probability assignments for $Y$. There is compatibility between $\mds{Q}_1$ and $\mds{Q}_2$ if, and only if, $p \neq q $ and $p,q \in (0,1)$. 
\label{Prop:CompatibilityForOut2}
\end{prop}

\begin{proof}
The proof follows directly from Theorem~\ref{Thm.NecessaryAndSufficientConditionForCompClassical}. As a matter of fact, suppose there is agreement between $\mds{Q}_1$ and $\mds{Q}_2$, then $\mbox{supp}[\mds{Q}_1(Y)] \cap \mbox{supp}[\mds{Q}_2(Y)] \neq 0.$ If either $p$ or $q$ were equal to one, then one of the cases would hold true:
\begin{itemize}
\item either $\mds{Q}_1=(1,0)$ and $\mds{Q}_2=(0,1)$
\item or $\mds{Q}_1=(0,1)$ and $\mds{Q}_2=(1,0)$.
\end{itemize}
In either case, Eq.~\eqref{Eq.ThmComptNecSuffClass} would prevent the existence of agreement. 

As for the other direction, suppose that $p,q \in (0,1)$. In this case, all the entries of the vectors $(p,1-p)$ and $(q,1-q)$ are non-vanishing. Once again, Thm~\ref{Thm.NecessaryAndSufficientConditionForCompClassical} implies that $\mds{Q}_1$ and $\mds{Q}_2$ are indeed compatible one another. 
\end{proof}

\begin{prop}
Let $Y$ be a discrete random variable. If $\mds{Q}_1(Y)$ is the uniform distribution for $Y$, then any other probability assignment $\mds{Q}_{2}(Y)$ is compatible with $\mds{Q}_1(Y)$.
\label{Prop:CompatibilityForUniform}
\end{prop}

\begin{proof}
Once again, the proof is a direct by-product of Thm.~\ref{Thm.NecessaryAndSufficientConditionForCompClassical}. It suffices to bear in mind that as 
\begin{align}
\mds{Q}_{1}=\left( \frac{1}{\vert \mbox{Out}(Y)  \vert},\frac{1}{\vert \mbox{Out}(Y)  \vert},...,\frac{1}{\vert \mbox{Out}(Y)  \vert} \right),
\end{align}
then the support of any other probability distribution must overlap with supp$[\mds{Q}_{1}(Y)]$. That completes the proof.
\end{proof}

Next, the quantum version generalizing Def.~\ref{Def.CompatibilityObjective} plus Def.~\ref{Def.CompatibilitySubjective} and Thm.~\ref{Thm.NecessaryAndSufficientConditionForCompClassical} is obtained via the hybrid states technique we discuss in Subsec.~\ref{Subsec.ConditionalQS}. Leveraging the concept of Bayesian condition we have got from this approach it is straightforward to come up with two other quantum definitions having the same interpretation as the classical we discussed in the previous paragraph, that is to say:

   \begin{dfn}[Objective Quantum Compatibility]
 Two assignments $\sigma_{B}^{1}$ and $\sigma_{B}^{2}$, for a quantum region $B$, are \emph{compatible} whenever it is possible to find two classical random variables $X_1$ and $X_2$, a hybrid state $\rho_{X_1X_2B}$, and two outcomes $x_1 \in \mbox{Out}(X_1)$ and $x_2 \in \mbox{Out}(X_2)$   such that:
 \begin{enumerate}
 \item $\rho_{X_1=x_1, X_2=x_2} > 0 $;
 \item $\sigma_{B}^{i}=\rho_{B \vert X_i=x_i}$, 
 \end{enumerate}
 for all $i \in \{1,2\}$. 
 \label{Def.CompatibilityObjectiveQuantum}
\end{dfn}

\begin{dfn}[Subjective Quantum Compatibility]
Two assignments $\sigma_{B}^{1}$ and $\sigma_{B}^{2}$, for a quantum region $B$, are \emph{compatible} whenever it is possible to find another classical random variable $X$ and a conditional state $\rho_{X \vert B}$ that both agents agrees upon such that:
\begin{enumerate}
\item $\mbox{Tr}_{B}\left( \rho_{X=x \vert B} \sigma_{B}^{i} \right)$ for all $x$ and for all $i \in \{1,2\}$;
\item  $\rho_{B \vert X=\tilde{x}}^{1} = \rho_{B \vert X=\tilde{x}}^{2}$, for some value $\tilde{x}$ in $\mbox{Out}(X)$, 
where each $\rho_{B \vert X=x}^{i}$ is given by the quantum Bayes' rule~\cite{LS14}
\begin{align}
\rho_{B \vert X=\tilde{x}}^{i}:= \frac{\rho_{X=x \vert B} \star \sigma_{B}^{i}}{\mbox{Tr}_{B}\left( \rho_{X=x \vert B} \sigma_{B}^{i} \right)}.
\label{Eq.DefQSubCompPi}
\end{align}
\end{enumerate}
\label{Def.CompatibilitySubjectiveQuantum}
\end{dfn}

As for the result connecting both quantum definitions we have:

\begin{thm}
$\sigma_{B}^{1}$ and $\sigma_{B}^{2}$ are compatible, either subjectively or objectively, if and only if 
\begin{align}
\mbox{supp}[\sigma_{B}^{1}] \cap \mbox{supp}[\sigma_{B}^{2}] \neq \emptyset,
\label{Eq.ThmComptNecSuffClass}
\end{align}
where $\cap$ denotes the geometric intersection.
\label{Thm.NecessaryAndSufficientConditionForCompQuantum}
\end{thm}

\subsubsection{State Pooling}\label{Subsub.StatePooling}

Moving on, and getting back to the pooling task, to come up with a state that all the agents can agree on as being the assignment reflecting the views of the group as whole, we may use --once again-- the results of Ref.~\cite{LS14}. 

For doing so, we will restrict ourselves to the case of two agents, say Theo and Wanda,  and assume that initially these  two agents agree on assigning the same $\mds{P}(Y)$ to $Y$. Additionally, employing the objective perspective, we assume that the posterior difference between agents' assignments are due to having collected different data, so that we model these extra data by two random variables $X_1$ and $X_2$ accessed respectively by Theo and Wanda. In this scenario we can enunciate the following result: 

\begin{thm}
If a minimal sufficient statistics for $X_1$ w.r.t. $Y$ and a minimal sufficient statistics for $X_2$ w.r.t. to $Y$ are conditionally independent given $Y$, then the pooled state $\mds{Q}_{\mbox{pooled}}(Y)$ is given by
\begin{align}
\mds{Q}_{\mbox{pooled}}(Y)=c\frac{\mds{Q}_1(Y)\mds{Q}_2(Y)}{\mds{P}(Y)},
\label{Eq.ThmPooledClassical}
\end{align}
where $c$ is a normalization constant, independent of $Y$.
\label{Thm.ClassicalPooling}
\end{thm}

The quantum case is in complete analogy with the classical scenario, although now $\rho_B$ is the prior both Wanda and Theo assigns to a system of interest, and both $X_1$ and $X_2$ are the classical random variables representing the data the agents have acquired upon learning from different interactions with the system. With that in hands we can enunciate the following solution for the quantum version of the pooling task. 

\begin{thm}
If a minimal sufficient statistics $s_1$ for $X_1$ w.r.t. $B$ and a minimal sufficient statistics $s_2$ for $X_2$ w.r.t. to $B$ satisfy:
\begin{align}
\rho_{s_1(X_1)s_2(X_2) \vert B} = \rho_{s_1(X_1)\vert B}\rho_{s_2(X_2) \vert B},
\label{Eq.ThmMinimalSuffCondition}
\end{align} 
then the pooled state $\sigma_{B}^{pooled}$ is given by
\begin{align}
\sigma_{B}^{pooled}=c\sigma_{B}^{1} \rho_{B}^{-1} \sigma_{B}^{2},
\label{Eq.ThmPooledQuantum}
\end{align}
where $c$ is a normalization constant, independent of $B$.
\label{Thm.QuantumPooling}
\end{thm}

We emphasize that in either case Thm.~\ref{Thm.ClassicalPooling} and Thm.~\ref{Thm.QuantumPooling} make explicit how to combine distinct assignments arising from potentially different interactions with a system of interest. It is this recombination we will strongly use in the next section to say what should be the effective dynamics coming from a coarse-graining description.


\section{Emergent dynamics seen as agreement between agents}\label{Sec.MainResults}

Now we have got in hands all necessary ingredients to address our major topic. We have taken a long detour so far, so that it is pedagogical to restate our problem once again. In Ref.~\cite{DCBdM17} the authors sought what is the emergent dynamics, arising from a coarse-graining, compatible with the diagram in Fig.~\ref{Fig_Diagram_Quantum}. Although well-motivated, their notion of compatibility was highly algebraic, as it is expressed by demanding commutativity of Fig.~\ref{Fig_Diagram_Quantum}.

Here, on the other hand, we will frame the task of finding that emergent dynamics into decision theoretical world, and additionally will also interpret the compatibility demanded by the authors in Ref.~\cite{DCBdM17} through decision theoretical lenses. We feel this new layout is manifold, as it not only brings over a new perspective on the coarse-graining scenario, but also opens up new points of contact between quantum information and decision theory. On the top of that, the present work can also be seen as a direct application of quantum conditional states~\cite{LS14,LS13} formalism, which sets the problem addressed by the authors of Ref.~\cite{DCBdM17} also as an instance of (quantum) causal inference. 

Kicking it off, let us first re-draw the diagram depicted in Figure~\ref{Fig_Diagram_Quantum} in a way more adapted to our purposes. Figure ~\ref{Fig_Diagram_Agreement} captures the key points we want to emphasize. 

\begin{center}
    \begin{figure}
        \includegraphics[scale=0.2]{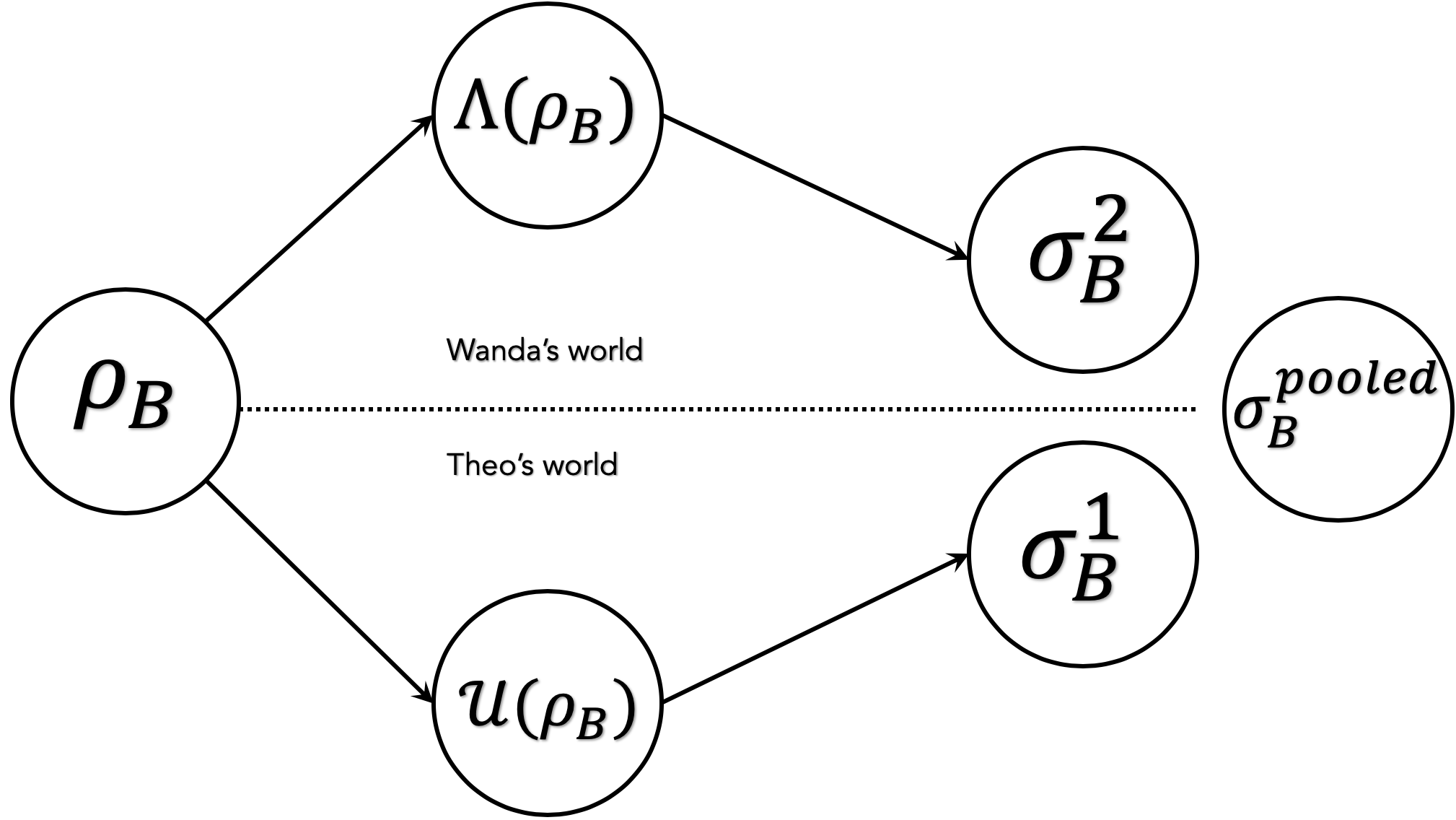}
        \caption{\begin{footnotesize} New coarse-graining diagram. The dotted line depicting a schematic distinction between what Wanda and Theo experience. The pooled state represents the state the two agents must agree on at the end of the process.\end{footnotesize} }
        \label{Fig_Diagram_Agreement}
    \end{figure}
\end{center}

There two agents, say Wanda and Theo, have once had the same prior, the same initial description, for a quantum system $\rho_{B}$. After that, they went on and potentially interacted differently with $B$. Think of it as modeling access to different experimental devices and potentially different measurements. Think of it as if Wanda has got access to a defective detector, so the data she collects represent an imperfect (and distinct from Theo's) description of the system. Time goes on, the system of interest then evolves, and now is Theo who is (potentially) affected by an imperfect measurement apparatus. Her final state $\sigma_{B}^{2}$ and his $\sigma_{B}^{1}$.

 They are trying to infer physical properties from the same system, undergoing the same evolution, through potentially distinct and imperfect detectors, though. When are their assigned states compatible one another? When can they reconcile their differences and assign one single state $\sigma_{B}^{pooled}$ --to the quantum region $B$ of interest-- representing accurately both agents' beliefs? Summing up, when is it the case that $\sigma_{B}^{1}$ and $\sigma_{B}^{2}$ agree with each other?  

It is not always the case, though, that such compatibility will take place. Using the algebraic formalism developed in Ref.~\cite{DCBdM17}, it is possible to engineer simple scenarios, with underlying unitary dynamics and very simple imperfections, where there is no (algebraic) compatibility. We believe the same may also happen here. On the other hand, what we want to come up with here are necessary and sufficient conditions for when such condition holds true. 

The way we have set the task up fits perfectly into the objective version of quantum state pooling  as expressed by Thm.~\ref{Thm.QuantumPooling}. As a matter of fact, it is the \emph{minimal sufficient statistics} conditions that control whether or not each agent's interaction with different data-sets will hinder them of assigning a single pooled state at the very end. In other words, it is the \emph{minimal sufficient statistics} that says whether defective apparatuses will strongly impact each agent's view about the system. If so, the data the agents gather are so distinct from each other that by not respecting the independence condition expressed in Eq.~\eqref{Eq.ThmMinimalSuffCondition} they cannot assigned a single state accurately representing their world's view.  

Our main result can be expressed as follows:

\begin{result}
If a minimal sufficient statistics $s_1$ for $X_1$ w.r.t. $B$ and a minimal sufficient statistics $s_2$ for $X_2$ w.r.t. to $B$ satisfy the independence condition in Eq.~\eqref{Eq.ThmMinimalSuffCondition},
then the pooled state $\sigma_{B}^{pooled}$ is given by
\begin{align}
\sigma_{B}^{pooled}=c\sigma_{B}^{1} \rho_{B}^{-1} \sigma_{B}^{2},
\label{Eq.MainResult}
\end{align}
where $c$ is a normalization constant, independent of $B$.
\label{Rslt:MainResult}
\end{result}

It turns out that it is the very Eq.~\eqref{Eq.MainResult} that allows us to express our main result as a mapping, and consequently in a closer analogy to what has been done in Ref.~\cite{DCBdM17}. For each initial prior $\rho_{B}$ there is only one $\sigma_{B}^{1}$ assigned by Wanda, and also only one $\sigma_{B}^{1}$ assigned by Theo. The particular details of how each assignment has been done does not matter. As long as they respect the minimality expressed by Eq.~\eqref{Eq.ThmMinimalSuffCondition}, they could have originated out of a channel, from Bayesian condition or by any other means. Additionally, if we can trace them back to $\rho_{B}$ (for expressing them as a function) while respecting the independence-like condition in Eq.~\eqref{Eq.ThmMinimalSuffCondition}, our result says that it is possible to define a map from $\mc{D}(\mc{H}_{B})$ onto itself given by:
\begin{align}
\tilde{\Gamma}: \rho_{B} \mapsto c\sigma_{B}^{1} \rho_{B}^{-1} \sigma_{B}^{2},
\label{Eq.MainResultAsChannel}
\end{align}
where $c$ is a normalization constant independent of $B$. Due to the format of Eq.~\eqref{Eq.MainResultAsChannel} we should emphasize that although being well-defined, the assignment map $\tilde{\Gamma}$ will be non-linear in general. Not only the inverse of matrix does not distribute across the sum, but $\sigma_{B}^{1,2}$ is also highly dependent of the prior $\rho_{B}$.

Note that whereas in Ref.~\cite{DCBdM17} the authors are focused on getting a quantum channel as the emergent dynamics arising from the coarse-graining process, in here on the other hand, the map we have just defined is not restrict by such constraints. Within our novel framework,  it does not matter whether the assigned pooled state $\sigma_{B}^{pooled}$  can be seen as coming out of quantum channel~\cite{NC00,MichaelGuide,WatrousLect}. Our  notion of compatibility has been expressed via decision-theoretical arguments, and as long as the assigned state represents accurately the agents' perspective, the formalism works and we do not have to take into account if this assignment has arisen from a completely positive trace preserving map~\cite{NC00,MichaelGuide,WatrousLect}. This opens up another route, with more freedom, to address the coarse-graining problem.


\section{Conclusion}\label{Sec.Conclusion}

Trying to come up with a simple, mathematically rigorous toy model to explain the emergence of non-quantum dynamics, the authors of Ref.~\cite{DCBdM17} explored the paradigm of coarse-graining, as diagrammatically depicted in Fig.~\ref{Fig_Diagram_Quantum}. There they demanded that the emergent, effective dynamics should be compatible with the underlying diagram. The emergent dynamics should be such that the diagram commutes. 

Putting aside the quantum channel language, and implementing a decision theoretical perspective, in this contribution we switched the notion of compatibility. Instead of asking for something purely algebraic, we demanded the compatibility between different descriptions were given in terms of a sort of agreement between distinct probability assignments. Broadly, whenever both agents agree on a description accurately representing the information they have gathered about a system of interest over the time, we say that our compatibility notion has been fulfilled. On top of that, we also say that the state they have agreed upon is the coarse-grained state. This latter understood as if arising as an output and out of an emergent map, as shown in Eq.~\eqref{Eq.MainResultAsChannel}. Our work, therefore, being a combination of decision-theoretical ideas together with the conditional quantum state approach. A path opened up by the authors of~\cite{LS14} that we wanted to apply to different situations, to get different results.

Within the original coarse-graining scenario the authors pursed the path of trying to obtain necessary and sufficient conditions for the emergence of an effective dynamics. Although they did end up proving conditions for the existence of such maps, their result is not fully satisfactory, as it relies on checking out an infinite number of semi-definite programs~\cite{Lovasz95,CS17,BV04}. On the other hand, within our formalism, although we have departed from the usual quantum information parlance, Result~\ref{Rslt:MainResult} shows how it was possible to get more meaningful necessary and sufficient conditions for the existence of a quantum state representing the ideal final state arising from a coarse-graining process. 

The minimal sufficient statistics condition we demand there may be seen as nothing but an independence requirement, much in the molds of what is usually done in causality~\cite{Pearl2013Book,Kleinberg2013,Pearl1995,BP16,LS13}. Whenever such condition holds true, we can factorize the pooled state $\sigma_{B}^{pooled}$  decomposing it in a product-like form using as building blocks the agents' descriptions $\sigma_{B}^{1}, \sigma_{B}^{2}$ and the shared prior $\rho_{B}$.

Although this relaxation escapes a bit from the usual quantum-information paradigm, decision theoretical arguments are not new in the area and in conjunction with the conditional quantum states toolkit they have already brought over new insights on old problems~\cite{LS13,KG16,LS14,LP08,BLO19,LJBR12,CS16,FR18,LP17} we have faced in the field. 

Exploring this synergy created by compatibility and conditional quantum states formalism, connections with causality, steering and more general non-locality scenarios are topics that should be explored elsewhere.  In particular, considering only the toolkit arising from the conditional states, we might also have approached the original coarse-graining problem from the usual perspective, so that the emergent quantum channel would be defined through the inversion of the first vertical arrow in the diagram (see Fig.~\ref{Fig_Diagram_Quantum}). We are already taking care of this case in another work, though.

 We  feel this new framework is not only more intuitive, but also provides more tangible necessary and sufficient conditions for when it is possible to assign a meaningful, accurate, effective  state (assumed here to be the pooled state) arising as output of the coarse-graining problem.


\begin{acknowledgments}
CD thanks Fernando de Melo,  Tha\'{i}s Matos Ac\'{a}cio and Matt Leifer for all 
valuable discussions. 
 CD has been supported by a fellowship from the Grand Challenges 
Initiative at Chapman University. 
\end{acknowledgments}

\bibliography{biblio}
\end{document}